\documentclass[letterpaper,journal]{IEEEtran}
\usepackage{amsmath,amsfonts}
\usepackage{array}
\usepackage{textcomp}
\usepackage{stfloats}
\usepackage{url}
\usepackage{verbatim}
\usepackage{graphicx}
\usepackage{cite}
\usepackage{multirow}
\usepackage{pifont}
\usepackage{subcaption}
\usepackage{setspace}
\usepackage{makecell}
\usepackage{booktabs}   
\usepackage{enumitem}
\usepackage[most]{tcolorbox}
\newcommand{\xwx}[1]{\textcolor{black}{#1}}
\definecolor{darkblue}{rgb}{0, 0.0, 0.78}
\usepackage{xcolor}
\usepackage{listings}
\definecolor{bg}{RGB}{246,247,245}
\definecolor{frame}{RGB}{160,170,160}  
\definecolor{title}{RGB}{120,135,125}  
\definecolor{font}{RGB}{80,90,85}
\usepackage{amsthm}
\usepackage{amssymb}
\newtheorem{theorem}{Theorem}[section]  
\newtheorem{corollary}{Corollary}[theorem]
\tcbset{
  morandibox/.style={
    enhanced,
    breakable,
    colback=bg,
    colframe=frame,
    colbacktitle=title,
    coltitle=white,
    fonttitle=\bfseries\large,
    sharp corners=all,
    boxrule=1pt,
    title={Prompt for System Agent},
    attach boxed title to top left={yshift=-2mm, xshift=3mm},
    boxed title style={
      sharp corners=all,
      boxrule=0pt,
      frame hidden,
      interior style={fill=title}
    }
  }
}
\usepackage[ruled,vlined]{algorithm2e}

\begin{document}

\title{Towards Reliable Agentic Progressive Text-to-Visualization with Verification Rules}

\author{
    \IEEEauthorblockN{Wenxin Xu\IEEEauthorrefmark{1}, Chen Jason Zhang\IEEEauthorrefmark{1}, Xiaoyong Wei\IEEEauthorrefmark{1}, Haoyang Li\IEEEauthorrefmark{1}, Hwanhee Kim\IEEEauthorrefmark{1}, \\ 
    Yuanfeng Song\IEEEauthorrefmark{2}, 
    Raymond Chi-Wing Wong\IEEEauthorrefmark{3}\\
    }
    \IEEEauthorblockA{\IEEEauthorrefmark{1}The Hong Kong Polytechnic University, Hong Kong, China
    \IEEEauthorrefmark{2} ByteDance, Shanghai, China
    \\ \IEEEauthorrefmark{3} The Hong Kong University of Science and Technology, Hong Kong, China
    }
}

\markboth{Journal of \LaTeX\ Class Files,~Vol.~14, No.~8, August~2021}%
{Shell \MakeLowercase{\textit{et al.}}: A Sample Article Using IEEEtran.cls for IEEE Journals}


\maketitle

\begin{abstract}
Text-to-Visualization (Text-to-Vis) translates natural language queries into visualization query languages, enabling non-expert users to perform data analysis. However, most existing methods follow a one-shot paradigm that requires users to specify all visualization details in a single round, often leading to cognitive overload and incorrect visualizations. In this paper, we propose PMVis, a progressive multi-turn paradigm for text-to-vis, where users' intents are refined through multi-turn interactions. To support research in this paradigm, we construct PMVisBench, the first dataset designed to capture the progressive and iterative nature of real-world user queries. It is built through VQL simplification and NLQ reconstruction, with explicit rule constraints to ensure each intermediate VQL remains valid and meaningful. Building upon PMVis, we further introduce PMVisAgent, an agent-based framework that simulates realistic user–system dialogues. PMVisAgent consists of a User, a System, and a Validation Agent that performs verification and repair via a ReAct-style tool-use loop to mitigate error accumulation across rounds, with explicit interaction and verification rules to ensure reliability of the multi-agent system. Extensive experiments on PMVisBench demonstrate that PMVisAgent significantly outperforms state-of-the-art text-to-vis baselines. It achieves up to 17.57\% and 23.21\% improvements in execution accuracy in single-table and multi-table settings, respectively, while ablation studies confirm the importance of combining progressive interaction with clarification. 
The code is available at \url{https://github.com/wxxv/PMVis}.
\end{abstract}


\section{Introduction}
Data visualization (DV) plays an important role in modern data management systems by serving as an intuitive interface between users and structured data. By mapping query results to visual encodings, DV helps users inspect distributions, compare groups, identify trends, and communicate data insights. As visualization is closely related to query formulation, result interpretation, and exploratory data analysis, it has been extensively studied in the database community~\cite{luo2021synthesizing,bhowmick2020aurora,Lu2026MultiVis,xie2024haichart,ji2024navigating,vartak2015seedb,chai2020crowdchart,luo2020steerable,zhang2024natural,ji2023visualization,lu2025towards,wan2025datavist5}. However, constructing effective visualizations often requires users to specify their intent using declarative visualization languages, such as Vega-Lite~\cite{satyanarayan2016vega}, ggplot2~\cite{villanueva2019ggplot2}, and Vega-Zero~\cite{luo2021natural}. Despite their expressiveness, these languages require specialized syntax and therefore impose a considerable barrier on non-expert users.

\begin{table*}[t]
    \caption{Comparison of text-to-vis benchmarks.}
    \label{tab:benchmark_comp}
    \centering
    \resizebox{\linewidth}{!}{%
    \begin{tabular}{lccccccccl}
    \toprule
    \textbf{Dataset} & \textbf{\#-Tables} & \textbf{\#-NL} & \textbf{\#-VIS} & \textbf{\#-NL/\#-VIS} & \textbf{\#-Chart Types} & \textbf{Multi-turn} & \textbf{Progressive} & \makecell{\textbf{NL Query}\\\textbf{Generation}} \\
    \midrule
    nvBench~\cite{luo2021synthesizing}         & 780  & 25750  & 7247  & 3.55 & 4 & \ding{55} & \ding{55} & Rule-based \\
    nvBench 2.0~\cite{luo2025nvbench}          & 780  & 7878  & 24076 & 0.33 & 6 & \ding{55} & \ding{55} & LLM-based \\
    VisEval~\cite{chen2024viseval}             & 748  & 2524  & 1150  & 2.19 & 4 & \ding{55} & \ding{55} & LLM-based \\
    nvBench-Rob~\cite{lu2025towards}           & 552  & 1182  & 1182  & 1.00 & 7 & \ding{55} & \ding{55} & LLM + Human \\
    NVBench-Feedback~\cite{xiong2025interactive} & 837   & 23225 & 6327  & 3.67 & 7 & \ding{55} & \ding{55} & LLM-based \\
    Dial-NVBench~\cite{song2024marrying}       & 780   & 21725 & 4495  & 4.83 & -- & \ding{51} & \ding{55} & Rule-based \\
    \midrule
    \textbf{PMVisBench (Ours)}                 & 748  & 2523  & 1149  & 2.20 & 4 & \ding{51} & \ding{51} & LLM + Rule-based + Human \\
    \bottomrule
    \end{tabular}%
    }
\end{table*}

To lower the barrier of visualization construction, text-to-visualization (text-to-vis) has emerged as an important task in the database community~\cite{luo2021synthesizing,song2022rgvisnet}. Given a natural language query (NLQ), a text-to-vis system translates the user intent into a visualization query language (VQL), a visualization-oriented intermediate representation that captures chart semantics such as visual marks, data attributes, filters, aggregations, and grouping operations. The generated VQL can then be compiled into a declarative visualization language (DVL) specification to render the final chart. This pipeline frees users from manually writing DVL code and enables non-expert users to create visualizations through natural language interaction.

Prior work has made important progress in text-to-vis by developing large-scale benchmarks, stronger evaluation protocols, and more capable generation models. Representative efforts include nvBench~\cite{luo2021synthesizing}, VisEval~\cite{chen2024viseval}, nvBench 2.0~\cite{luo2025nvbench}, nvBench-Rob~\cite{lu2025towards}, and recent LLM- or agent-based frameworks for visualization generation~\cite{wu2024automated,ouyang-etal-2025-nvagent}. These studies cover general NLQ-to-VQL translation, data quality, linguistic robustness, ambiguous-query understanding, and complex multi-table reasoning. However, despite this progress, they predominantly evaluate text-to-vis as a one-shot translation problem: the system receives a fully specified NLQ and is judged by whether it can directly generate the corresponding final VQL. 

\noindent\textbf{Challenges.} The one-shot formulation overlooks a common property of real exploratory analysis: users rarely specify all visualization requirements (e.g., chart type, filters, aggregation, and grouping) at once, but instead start with a coarse, high-level request and progressively refine it through follow-up rounds. Forcing such interactions into a single fully specified query increases cognitive burden, especially for non-experts, and frequently produces visualizations that deviate from users' true intent. A more realistic setting is therefore \emph{progressive multi-turn refinement}, where the system generates a valid, renderable visualization at every intermediate turn and incrementally refines it as new user requirements arrive. This raises a critical question: \emph{how can we extend text-to-vis from one-shot translation to progressive interaction while maintaining valid visualizations at each turn?}

\noindent\textbf{Existing Benchmarks and Their Limitations.} As summarized in Table~\ref{tab:benchmark_comp}, existing text-to-vis benchmarks reflect this one-shot assumption. They cover general translation, evaluation quality, linguistic robustness, and query ambiguity through nvBench~\cite{luo2021synthesizing}, VisEval~\cite{chen2024viseval}, nvBench-Rob~\cite{lu2025towards}, and nvBench 2.0~\cite{luo2025nvbench}, and are suitable when the goal is to evaluate the final fully specified visualization query. However, they are not designed to evaluate the more realistic setting of progressive multi-turn interaction, where users begin with coarse analytical intents and incrementally add constraints over several rounds. Existing datasets provide only the terminal NLQ--VQL pair and omit intermediate user refinements and executable visualization targets at each turn. As a result, they cannot support the study of whether a system can preserve valid visualizations throughout an interaction, correctly incorporate newly added constraints, or avoid error accumulation across turns. The closest multi-turn effort, Dial-NVBench~\cite{song2024marrying}, provides conversational dialogues for VQL-related questions; however, its final VQL is still determined only by the last complex query, leaving intermediate turns without executable visualization targets. Consequently, directly using existing datasets would reduce the problem to ordinary one-shot translation and miss the central behavior that we aim to study. Thus, prior benchmarks lack the key ingredients needed for progressive refinement: structurally organized refinement trajectories and turn-by-turn valid VQL generation. \emph{This gap leaves progressive multi-turn text-to-vis underexplored as a task paradigm and benchmark setting.}

\noindent\textbf{Our proposal.} To address this limitation, we propose \textbf{PMVis}, a novel progressive multi-turn paradigm for text-to-vis. Realizing this paradigm presents two key challenges: (1) the lack of benchmarks that capture realistic, progressive refinement trajectories, and (2) error accumulation across multi-turn translation. 

For the first challenge, we construct \textbf{PMVisBench} to instantiate this progressive multi-turn formulation. PMVisBench is not intended to replace one-shot benchmarks; instead, it is a trajectory-level extension built on VisEval~\cite{chen2024viseval}. We treat each original fully specified NLQ--VQL pair as the terminal state of a progressive session and derive preceding intermediate turns that gradually introduce visualization requirements. As shown in Table \ref{tab:benchmark_comp}, each session is organized as a \emph{refinement trajectory}, where every intermediate round produces a valid, executable VQL that renders a meaningful visualization, and the following rounds progressively add specification details, forming a natural, simple-to-complex progression that mirrors real-world exploratory analysis. To avoid assuming a single deterministic refinement path, PMVisBench uses randomized clause masking to generate diverse plausible trajectories, thereby simulating the uncertainty of how users may progressively refine their intents. During this process, we enforce explicit constraints to prevent invalid simplifications. Specifically, we apply (i) \emph{Masking-Constraint rules} to preserve core clauses and clause/column dependencies, and (ii) \emph{Visualization-Feasibility rules} to ensure each intermediate VQL remains renderable with non-empty results. We then generate corresponding NLQs with LLMs, followed by a comprehensive manual correction for linguistic naturalness and semantic alignment.

For the second challenge, in progressive multi-turn translation, the system builds VQL of each round upon previous outputs. Once an early round introduces errors, it propagates through the dialogue history and compounds in subsequent rounds, ultimately degrading the final visualization. To tackle this error accumulation problem, we introduce \textbf{PMVisAgent}, an agent-based framework that conducts user--system dialogues with rule-guided verification and repair. PMVisAgent comprises (i) a \emph{User Agent} that issues NLQs and provides minimal clarifications when needed, (ii) a \emph{System Agent} that translates the current context into a candidate VQL, and (iii) a \emph{Validation Agent} that applies a ReAct-style loop with tools to validate and iteratively refine the candidate. We further define explicit rules for the User and Validation agents to prevent information leakage during clarification and to regulate tool usage and repair orders. We also prove that these rules provide PMVisAgent with theoretical reliability guarantees. Together, these designs enable progressive refinement and improve end-to-end correctness and robustness in multi-turn settings.

Extensive experiments on PMVisBench demonstrate that PMVisAgent significantly outperforms existing one-shot text-to-vis models. Against state-of-the-art baselines, PMVisAgent improves execution accuracy by 17.57\% in single-table and 23.21\% in multi-table settings. Ablation studies further show clear performance drops when removing either the multi-turn interaction or validation module, while the full framework consistently performs best, confirming the effectiveness of our progressive multi-turn validation paradigm.

\noindent\textbf{Contributions.} Our main contributions are summarized as follows:
\begin{itemize}
    \item We introduce an innovative progressive multi-turn text-to-vis paradigm, PMVis, and construct PMVisBench, the first benchmark dataset designed to capture iterative real-world refinement of user queries across diverse scenarios.
    \item We propose PMVisAgent, an agent-based framework that realizes the PMVis paradigm. The framework leverages multi-agent collaboration and incorporates a ReAct-style reasoning loop with tool use to mitigate error accumulation across rounds, thereby producing more accurate and robust visualizations.
    \item Extensive experiments on the PMVisBench demonstrate that PMVisAgent significantly outperforms existing text-to-vis models based on the traditional one-shot paradigm, achieving up to 88.6\% and 77.86\% execution accuracy in single-table and multi-table settings. Ablation studies further reveal that the full multi-turn validation framework consistently delivers the best results, confirming the effectiveness of our proposed approach.
\end{itemize}

\section{Related Work}
\subsection{Benchmark Datasets for Automatic Data Visualization}

Recent progress in automatic visualization spans NLP, data mining, and database systems \cite{ge2024automatic, song2022rgvisnet, qian2021learning, luo2021synthesizing, tang2022sevi}, motivating diverse benchmarks.
Luo et al.\cite{luo2021synthesizing} proposed NVBench, synthesizing text-to-vis examples via a unified AST. Chen et al. \cite{chen2024viseval} introduced VisEval and an automated evaluation pipeline. Lu et al. released nvBench-Rob \cite{lu2025towards} to test robustness under lexical/phrasal perturbations.
Song et al. \cite{song2024marrying} extended benchmarks to conversational settings with CoVis and Dial-NVBench, and nvBench 2.0 \cite{luo2025nvbench} further targets ambiguity with multiple valid visualizations.

Despite addressing robustness, dialogue, and ambiguity, most benchmark datasets still followed the one-shot paradigm, even in conversational settings \cite{song2024marrying} where the system typically outputs the final visualization at the final round.
In contrast, we introduce the PMVis paradigm and construct PMVisBench to evaluate progressive multi-turn refinement, aligning with how users naturally clarify and refine intents across rounds.

\subsection{Text-to-Visualization (Text-to-Vis)}

Text-to-vis translates NL queries into executable visualization specifications, reducing the reliance on templates or manual configuration \cite{gao2015datatone, hoque2017applying, yu2019flowsense}.
Early approaches based on traditional neural networks include Data2Vis \cite{dibia2019data2vis}, ncNet \cite{luo2021natural}, and retrieval-augmented RGVisNet \cite{song2022rgvisnet}.
With the rapid development of LLMs, prompt-based methods such as Chat2Vis \cite{maddigan2023chat2vis} and Prompt4Vis \cite{li2025prompt4vis} substantially improve the quality of visualizations, while nvAgent \cite{ouyang-etal-2025-nvagent} adopts a multi-agent pipeline to enhance reasoning over complex schemas.

These methods showcase the potential of LLMs for text-to-vis, but they still follow fixed workflows, where each component rigidly executes its role, sacrificing the flexibility of LLMs. Thus, we introduce PMVisAgent, an agent-based framework that supports the PMVis paradigm, introducing a ReAct-style validation step, enabling the LLM to self-plan, detect, and correct errors during the VQL generation process.

\subsection{Multi-turn Interactive Systems}

Multi-turn dialogue systems include open-domain agents (e.g., Microsoft XiaoIce\cite{zhou2020design}, Google Meena\cite{adiwardana2020towards}, Meta's BlenderBot\cite{shuster2022blenderbot}) and task-oriented systems that combine modular or end-to-end designs for robustness \cite{chen2018dialogue, zhang2020recent, balaraman2021recent, li2017end, wen2017network, yang2021ubar}.
In text-to-vis, Dial-NVBench \cite{song2024marrying} is a representative multi-turn benchmark, and MMCoVisNet addresses VQL-related questions, but visualization generation remains effectively one-shot, driven by the final complex query.

Our work focuses on progressive multi-turn interaction for text-to-vis. Unlike previous approaches that assume a one-shot generation of output VQL, PMVis supports incremental refinement across multiple rounds. This paradigm decomposes a complex query into simpler queries to reduce users’ cognitive burden and better reflect how users naturally explore and express their analytical intentions.

\section{Problem Formulation \& PMVisBench}

In this section, we first introduce the PMVis paradigm, highlighting the key characteristics that distinguish it from the traditional one-shot text-to-vis paradigm. Then we describe the construction of PMVisBench, a dedicated dataset designed for PMVis.

\subsection{Problem Formulation}
\subsubsection{Traditional One-shot Paradigm in Text-to-Vis}
Let $q$ be a natural language question (NLQ) and $v$ the corresponding visualization query language (VQL), which can be compiled into a declarative visualization language (DVL) specification. The traditional one-shot paradigm translates an NLQ into a VQL in a single round:
\begin{equation}\label{eq:traditional paradigm}
    v = f(q)
\end{equation}
where $f(\cdot)$ denotes the translation function. The system then executes $v$ over the database to obtain the visualization.

This formulation implicitly assumes that (1) \emph{Completeness}. The input NLQ $q$ must fully specify the visualization intent in a single query, including the chart type, data attributes, filters, and other relevant details. (2) \emph{Determinism}. Each NLQ $q$ is mapped to a unique, deterministic VQL, ignoring the possibility of multiple semantically valid visualizations. (3) \emph{Independence}. Each NLQ-VQL pair is processed independently, without leveraging context from prior interactions. However, real users often start with coarse intents and refine them iteratively, making the one-shot paradigm prone to ambiguity and omission.

\begin{figure*}[th]
    \centering
    \includegraphics[width=\textwidth]{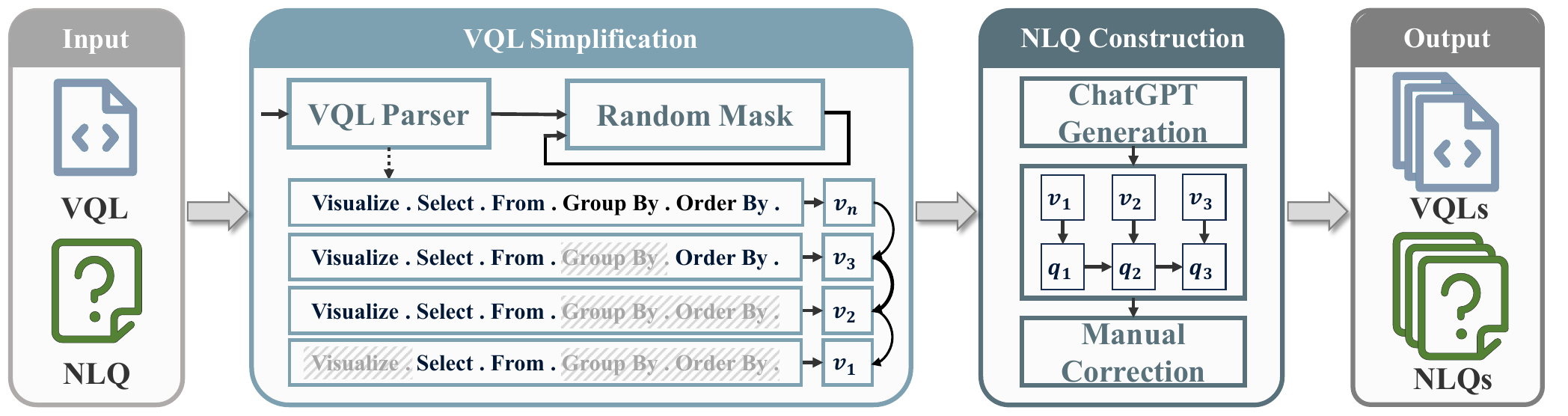}
    \caption{Pipeline of PMVisBench dataset construction. In the VQL Simplification step, original VQLs are parsed into clauses and progressively simplified through random masking while preserving mandatory components. Subsequently, in the NLQ Construction step, ChatGPT was employed to generate corresponding NLQs based on the simplified VQLs and dialogue history, after which entire dataset was manually inspected and corrected to ensure linguistic naturalness and semantic fidelity.}
    \label{fig:dataset_construction}
\end{figure*}

\subsubsection{PMVis Paradigm Formulation}
To align with iterative user behavior, we formulate \textit{progressive multi-turn} text-to-vis (PMVis) as a sequence of NLQ--VQL pairs:
\begin{equation}\label{eq:PMVis paradigm}
\mathcal{P}=\{P_i \mid P_i=(q_i, v_i),\ i=1,2,\dots,n\}
\end{equation}
where $q_i$ denotes the NLQ at round $i$, and $v_i$ represents the corresponding VQL generated by the system. PMVis is characterized by two essential properties that (1) \emph{Multi-turn}. Each pair ($q_i, v_i$) corresponds to one round of dialogue, capturing the iterative refinement process of the user. Unlike the independent one-shot mapping, later rounds explicitly depend on prior interactions.
(2) \emph{Progressive refinement}. At each round, the VQL is derived as:
\begin{equation}
    v_i = f(v_{i-1}, q_i)
\end{equation}
Accordingly, PMVis induces progressive sequences of NLQs and VQLs:
\begin{equation}\label{chain q}
    \mathcal{Q}=\{q_i \mid i=1,2,\dots,n\},\quad
    \mathcal{V}=\{v_i \mid i=1,2,\dots,n\}.
\end{equation}
Typically, $v_i$ contains strictly more semantic components (e.g., filters, grouping, aggregation) than $v_{i-1}$. The final round recovers the one-shot setting: $q_n=q$ and $v_n=v$ in Eq.~(\ref{eq:traditional paradigm}).

Notably, although CoVis~\cite{song2024marrying} involves multi-turn interactions, it generates visualizations only in the final round, which still follows the one-shot mapping in Eq.~(\ref{eq:traditional paradigm}). In contrast, PMVis decomposes a complex visualization intent into simplified sub-queries and produces an executable intermediate VQL at each turn as in Eq.~(\ref{eq:PMVis paradigm}).

\subsection{PMVis Dataset: PMVisBench}\label{sec:PMVisBench}

In this section, we describe the construction of \textbf{PMVisBench}, which is designed to support the PMVis progressive multi-turn paradigm.
As illustrated in Fig.~2, the overall pipeline consists of two steps: \emph{VQL simplification} and \emph{NLQ construction}.
Given an original complex NLQ--VQL pair $(q_n, v_n)$, we generate a progressive trajectory
$\langle(q_0,v_0),(q_1,v_1),\ldots,(q_n,v_n)\rangle$.
To ensure that each intermediate step is structurally legal and semantically valid, we constrain the simplification process with two families of rules:
\textbf{Masking Constraint (MC) rules} and \textbf{Visualization Feasibility (VF) rules}, denoted as \textbf{MC-Rule} and \textbf{VF-Rule}, respectively.
If a sampled mask violates any rule, we reject it and re-sample a new clause.

\subsubsection{VQL Simplification}

We first parse each $v_n$ into its constituent clauses (e.g., Visualize, Select, From, Where, Group By, etc.):
\begin{equation}
    C_n = VQLParser(v_n) = \{c_{sel}, c_{from}, c_1, c_2, \ldots, c_m\},
\end{equation}
where $c_{sel}$ and $c_{from}$ denote the \textit{Select} and \textit{From} clauses, and
$\{c_1,\ldots,c_m\}$ denote other optional clauses.

To create a progressive simplification path, we apply a random masking strategy that removes one optional clause per iteration:
\begin{equation}
    C_{i-1} = RandomMask(C_i) \triangleq C_i \setminus \{c\},\quad c \sim \texttt{Unif}(O_i),
\end{equation}
where $O_i = C_i \setminus \{c_{sel},c_{from}\}$ is the optional clause set at step $i$.
A sampled clause $c$ is accepted only if it satisfies the following rules.

\noindent\textbf{MC-Rule 1 (Non-maskable Core):} $c \notin \{c_{sel},c_{from}\}.$

\noindent\textit{Explanation.} We never mask the \textit{Select} and \textit{From} clauses to avoid structurally illegal queries.

\noindent\textbf{MC-Rule 2 (Clause Dependency):} $\forall c' \in (C_i\setminus\{c\}),\ \mathcal{P}(c') \subseteq (C_i\setminus\{c\}).$

\noindent\textit{Explanation.} Each clause keyword may require prerequisite keywords.
$\mathcal{P}(\cdot)$ returns the prerequisite clause set induced by VQL semantics.
For example, masking \textit{GroupBy} is forbidden if \textit{Having} remains; masking prerequisite aggregation context is forbidden if aggregation-dependent clauses remain.

\noindent\textbf{MC-Rule 3 (Column Dependency):} $\mathrm{Cols}(C_i\setminus\{c\}) \subseteq \mathcal{B}(C_i\setminus\{c\}).$

\noindent\textit{Explanation.} $\mathrm{Cols}(\cdot)$ extracts referenced columns appearing in the remaining clauses, while
$\mathcal{B}(\cdot)$ extracts columns still introduced (bound) by the remaining \textit{From/Join} clauses.
This rule rejects masking that removes bindings but leaves their variables referenced elsewhere (e.g., masking a \textit{Join} while columns from the joined table remain referenced).

\noindent\textbf{VF-Rule 1 (Renderable Visualization):} $(c_{vis}\in C_i\setminus\{c\}) \Rightarrow \mathcal{I}(C_i\setminus\{c\})=\texttt{true}.$

\noindent\textit{Explanation.} If the \textit{Visualize} clause exists, the remaining clauses must still form a renderable visualization intent, preventing ill-posed chart requests.

\noindent\textbf{VF-Rule 2 (Non-empty Table):} \emph{let} $E(C)=\mathrm{Exec}(\mathrm{Assem}(C))$, \emph{then} $|E(C_i\setminus\{c\})|>0.$

\noindent\textit{Explanation.} $E(C)$ denotes the result table obtained by executing the query assembled from clause set $C$.
We reject masks that lead to an empty result table.

Applying the above procedure repeatedly yields a reverse chain $C_0 \subset C_1 \subset \cdots \subset C_n$.
Each $C_i$ is deterministically assembled into a VQL $v_i$, forming the progressive refinement chain in Eq.~(\ref{chain q}).
Notably, the random masking mechanism avoids biasing the dataset toward a deterministic simplification pattern, while the MC/VF rules ensure that every intermediate step
remains valid and meaningful.

\subsubsection{NLQ Construction}

\begin{figure*}[t!]
    \centering
    \includegraphics[width=\textwidth]{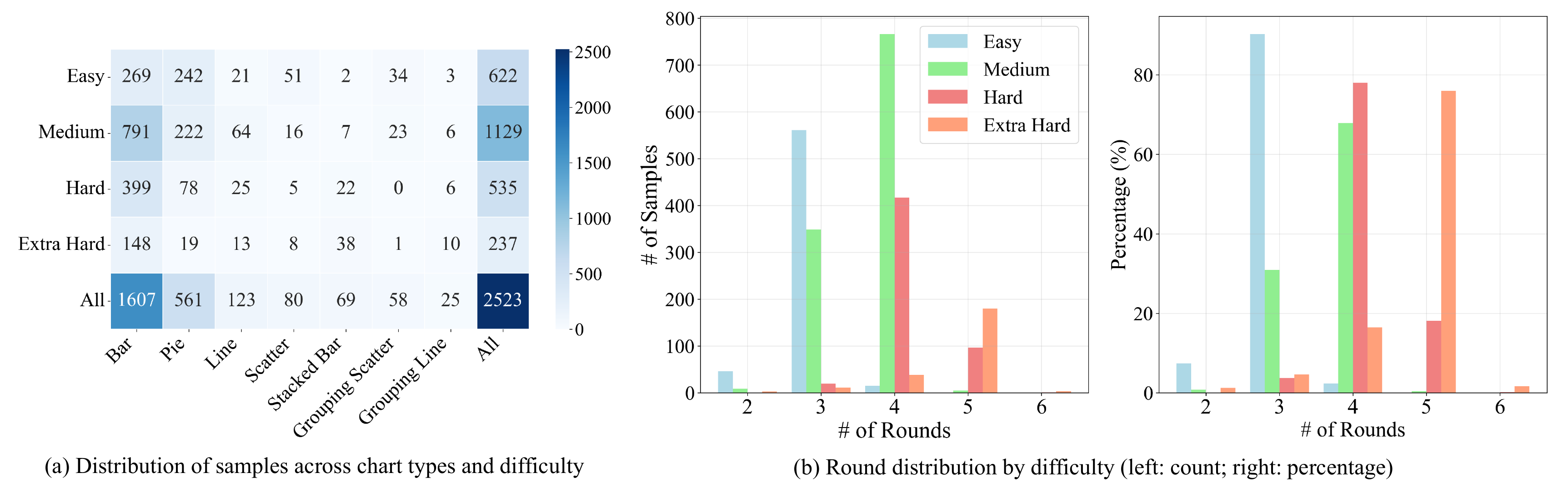}
    \caption{Dataset statistics of PMVisBench. (a) Distribution of samples across different chart types under four difficulty levels. (b) Distribution of rounds across difficulty levels, reported both in absolute counts (left) and percentages (right). Results indicate that easier cases usually require fewer rounds, while harder cases involve longer multi-turn refinement trajectories.}
    \label{fig:statistics}
\end{figure*}

After obtaining the sequence of simplified VQLs $\mathcal{V}=\{v_0,\ldots,v_n\}$, we construct the corresponding NLQs
$\mathcal{Q}=\{q_0,\ldots,q_n\}$ to form a progressive dialogue trajectory aligned with Eq.~(\ref{chain q}).
Concretely, leveraging ChatGPT, we first generate the initial NLQ $q_0$ corresponding to the most simplified VQL $v_0$.
Then, each subsequent NLQ $q_i$ is generated by considering both the dialogue history $\{P_1, \ldots, P_{i-1}\}$ and the semantic difference between $v_i$ and $v_{i-1}$.
This ensures that each follow-up question incrementally introduces one additional constraint (e.g., a filter, grouping, or aggregation), resembling how users refine intents in practice.

\textit{Quality Assurance}.
Since LLM outputs may be unstable and the progressive multi-turn trajectories in PMVisBench are obtained via reverse VQL simplification rather than forward real user interaction, we conduct comprehensive quality assurance over the entire dataset in three aspects: (i) NLQ quality and alignment, where each NLQ is manually corrected to be linguistically natural, semantically precise, and faithfully aligned with its corresponding VQL; (ii) trajectory plausibility, where we ensure it follows a realistic, reasonable, and user-like refinement pattern. For any sample that fails the human checks, we re-run the data generation script to regenerate the trajectory until it passes; and (iii) rule-based validity guarantee, where we enforce the MC-rule during VQL simplification, ensuring every intermediate step remains structurally legal and semantically valid. These procedures ensure the reliability and usability of PMVisBench, while also ensuring that the reverse-derived trajectories reasonably reflect real user behaviors.

\subsection{Statistics of PMVisBench}\label{sec:PMVisBench_Statistics}

To better understand the characteristics of PMVisBench, we present dataset statistics in Fig. \ref{fig:statistics}. Specifically, we obtain 1,149 distinct visualizations (VIS) and 2,523 (NL, VIS) pairs, covering 146 databases. Fig. \ref{fig:statistics}(a) shows the distribution of samples across different chart types and difficulty levels. The difficulty is categorized into four levels: Easy, Medium, Hard, and Extra Hard\xwx{, which are introduced in VisEval \cite{chen2024viseval}}. Fig.~\ref{fig:statistics}(b) further illustrates the distribution of dialogue rounds under different difficulty levels, reported both in absolute counts (left) and percentages (right). As expected, easier samples typically require only two or three rounds, while harder cases involve longer trajectories with up to five or six rounds. These statistics confirm that PMVisBench covers a wide range of complexity and interaction depth, making it well-suited for evaluating models under progressive multi-turn text-to-vis scenarios.

\begin{figure*}[ht!]
    \centering
    \includegraphics[width=\textwidth]{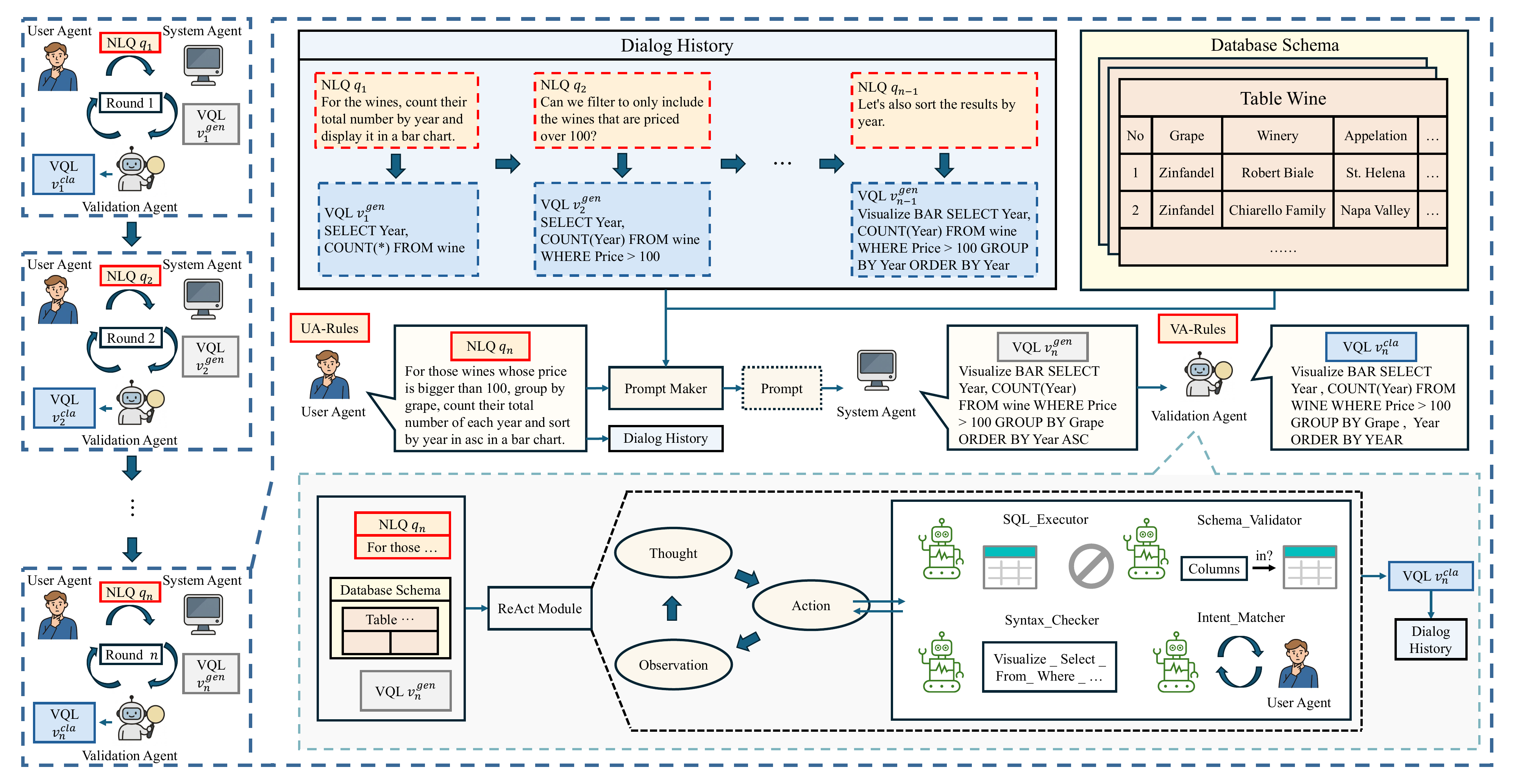}
    \caption{The overall architecture of PMVisAgent. The pipeline simulates a progressive multi-turn text-to-vis workflow with three collaborative agents. The $i$-th round text-to-vis transformation includes (i) a User Agent, which simulates real users by issuing NLQ $q_i$ and is responsible for clarifying visualization intents; (ii) a System Agent, which transforms the $q_i$ into its corresponding VQL $v_i^{gen}$ based on the dialogue history; and (iii) a Validator Agent, which employs a ReAct-style reasoning loop with tool usage (SQL executor, schema validator, syntax checker, intent matcher) to iteratively revise the $v_i^{gen}$ and provide the revised VQL $v_i^{cla}$, ensuring its accuracy in each round to reduce error accumulation across rounds.}
    \label{fig:PMVisAgent}
\end{figure*}

\section{PMVisAgent: An Agent-based Framework for PMVis Paradigm}

Building upon the PMVis paradigm, we introduce PMVisAgent, an agent-based framework designed to tackle the progressive multi-turn text-to-vis task. Unlike traditional approaches based on the one-shot paradigm, PMVisAgent explicitly models the progressive iterative nature of user–system interaction while mitigating the risk of error accumulation across rounds. As illustrated in Fig.~\ref{fig:PMVisAgent}, PMVisAgent involves three collaborative agents: a \textit{User Agent}, a \textit{System Agent}, and a \textit{Validation Agent}. 
We also define explicit rules for the User and Validation Agents and demonstrate that these rules provide theoretical reliability guarantees for PMVisAgent.
Through this collaborative workflow, PMVisAgent achieves both progressive refinement and robust correctness, enabling more faithful alignment with real user behaviors.

\subsection{User Agent}\label{sec:user-agent}

The User Agent serves as the proxy for users in the PMVisAgent framework. Its primary role is to issue an NLQ at each dialogue round and clarify the intent clarification request from the Validation Agent. Instead of independently determining the refinement path, the User Agent follows the pre-defined sequence of NLQs $\mathcal{Q}$ provided in PMVisBench, which ensures that the simulated interaction faithfully reflects realistic user behavior.

\textit{Dialogue state.}
For each dialogue instance, let the NLQ sequence be
$\mathcal{Q}=\langle q_1,\dots,q_n\rangle$.
Define $\mathcal{H}_i$ as the dialogue history up to round $i$ and $\mathcal{H}_0=\emptyset$, and define $\mathcal{S}_i$ as the dialogue status fed to the System Agent at round $i$.

\textbf{(i) NLQ issuance.} Define the issuance mapping $\eta:\mathbb{N}\times \mathcal{Q}\rightarrow \mathcal{Q}$:
\begin{equation}
    \eta(i,\mathcal{Q}) \;=\; q_i.
\end{equation}

\textbf{(ii) Status construction.} Define the status construction mapping $u:\mathcal{H}\times \mathcal{Q}\rightarrow \mathcal{S}$:
\begin{equation}
    \mathcal{S}_i \;=\; u(\mathcal{H}_{i-1}, q_i) \;=\; \mathcal{H}_{i-1}\cup \{q_i\}.
\end{equation}

\textbf{(iii) History update after validation.} After the Validation Agent outputs the clarified VQL $v^{cla}_i$, the history is updated as:
\begin{equation}
    \mathcal{H}_i \;=\; \mathcal{S}_{i-1}\cup \{v^{cla}_i\}.
\end{equation}

\textit{User Agent (UA) rules.}
To ensure user simulation is well-defined and prevent information leakage, we impose following UA rules.

\noindent\textbf{UA-Rule 1 (History Minimality): }
$
\forall z\in\mathcal{Z},\;\forall i,\quad z\notin \mathcal{H}_i.
$

\noindent\textit{Explanation.} Let $\mathcal{Z}$ denote the set of all intermediate system-side histories produced during validation (e.g., tool calls, error reports, reasoning traces).
The user-visible history only stores past NLQs and their finalized clarified VQLs and no intermediate process is ever recorded in $\mathcal{H}_i$.

\noindent\textbf{UA-Rule 2 (Non-leaking Clarification): }
During validation, the User Agent may be queried for clarification. Define $\Phi$ as the set of clarification requests, and $v_i^*$ as the ground-truth VQL at round $i$.
Define a clarification mapping:
\begin{equation}
\mathrm{Clarify}: \Phi \times \mathcal{H}_{i} \times v_i^* \times \mathcal{D}
\;\rightarrow\;
\mathcal{Q}^{clr},
\end{equation}
where $\mathcal{D}$ denotes the database schema, $\mathcal{Q}^{clr}$ denotes the space of natural-language clarification utterances.  
Define two predicates:
(i) $\mathrm{Sat}(\phi, q^{clr})\rightarrow\{0,1\}$ indicating that $q^{clr}$ resolves the request in $\phi$;
(ii) $\mathrm{Leak}(q^{clr}, v^*)\rightarrow\{0,1\}$ indicating that $q^{clr}$ reveals executable details of $v^*$ (e.g., exact VQL clauses). UA-Rule 2 enforces:
\begin{align}
\forall \phi_i\in\Phi,\;
q_i^{clr} \setminus\{&\bot\} =\mathrm{Clarify}(\phi_i,\mathcal{H}_i,v_i^*,\mathcal{D}) \notag\\
&\;\Rightarrow\;
\mathrm{Sat}(\phi_i,q_i^{clr})=1
\;\land\;
\mathrm{Leak}(q_i^{clr}, v_i^*)=0.
\end{align}
where $\bot \in \mathcal{Q}^{clr}$ denotes a refusal to provide clarification.

\noindent\textit{Explanation.} UA may consult $v_i^*$ to locate the ambiguous component, but must answer in natural language without exposing the exact ground-truth VQL.

\noindent\textbf{UA-Rule 3 (Rejection Mechanism): }
Define a predicate $\mathrm{AskGT}:\Phi\rightarrow\{0,1\}$ that flags requests attempting to extract ground-truth executable content (e.g., ``give me the exact VQL'').
UA-Rule 3 requires:
\begin{equation}
\forall \phi_i\in\Phi,\;
\mathrm{AskGT}(\phi_i)=1
\;\Rightarrow\;
\mathrm{Clarify}(\phi_i,\mathcal{H}_i,v_i^*,\mathcal{D})=\bot.
\end{equation}
\textit{Explanation.} If the Validation Agent requests the ground-truth VQL/SQL directly, the User Agent must refuse, ensuring zero leaking of $v_i^*$.

\textit{Remark.}
With UA-Rule 1--3, the User Agent remains a faithful simulator of progressive user behavior while preventing information leakage.

\subsection{System Agent}

The System Agent is responsible for translating NLQ into a candidate VQL. At the $i$-th round, it receives the dialogue status $\mathcal{S}_i$, which contains the dialog history of the previous $i-1$ rounds $\mathcal{H}_{i-1}$ and the new follow-up NLQ $q_i$, together with the database schema $\mathcal{D}$. Based on these inputs, the agent generates a preliminary VQL $v_i^{gen}$:
\begin{equation}
    v^{gen}_i = \mathcal{F}(\mathcal{S}_i, \mathcal{D}),
\end{equation}
where $\mathcal{F}(\cdot)$ denotes the LLM-based translation process. The database schema $\mathcal{D}$ is injected into the prompt to restrict the generation space, ensuring that the generated VQL $v_i^{gen}$ is grounded in the target database.

\begin{tcolorbox}[morandibox, title={Prompt for System Agent}]
\textbf{\#\#\# Task}\\
\# Given a \textcolor{font}{Natural Language Question} and \textcolor{font}{conversation context}, 
generate VQL based on their corresponding \textcolor{font}{Database Schemas}.

\textbf{\#\#\# Database Schemas:}\\
\textcolor{font}{\{Database Schema $\mathcal{D}$\}}

\textbf{\#\#\# Please follow the VQL Format Guidelines}\\
\textcolor{font}{\{VQL Format Guidelines\}}

\textbf{\#\#\# Natural Language Question}\\
\#\# \textbf{Round} \textcolor{font}{\{round\_id\ $i$\}}\\
\# \textbf{User:} \textcolor{font}{\{NLQ $q_i$\}}\\
\# \textbf{System:} \textcolor{font}{[Output VQL]}

\textbf{\#\#\# Output}\\
Generate only the VQL, no explanation. Format your response as:\\[1mm]
'''VQL

[Your VQL query here]\\
'''
\end{tcolorbox}

The prompt template used in the $i$-th round of dialog in the System Agent is shown above. It explicitly specifies the agent’s task, provides the relevant database schema $\mathcal{D}$, and enforces a strict output format for the generation of VQL $v_i^{gen}$ based on the input NLQ $q_i$.

\subsection{Validation Agent}\label{sec:validation-agent}

The Validation Agent is designed to validate and repair the candidate VQL $v^{gen}_i$ generated by the System Agent, and output a clarified VQL $v^{cla}_i$. It follows a ReAct-style loop, iteratively producing an internal thought, calling a tool action, receiving an observation, and updating the candidate VQL optionally until termination.

\textit{ReAct loop.}
At validation step $t$, the  Validation Agent maintains:
\begin{equation}
k_i^{(t)} \;=\; \{v^{(t)}_i,\ \mathcal{S}_i\}
\end{equation}
Specifically, the initial $k_i^{(0)}=\{v^{gen}_i,\mathcal{S}_i\}$. It performs a ReAct loop as:
\begin{align}
\textbf{(Thought)}\quad &\tau_t \;=\; \Omega\!\big(k_i^{(t)};\mathcal{D}\big), \label{eq:va-thought}\\
\textbf{(Action)}\quad  &a_t \;=\; \Pi\!\big(k_i^{(t)},\tau_t;\mathcal{D}\big), \label{eq:va-action}\\
\textbf{(Observation)}\quad    &o_t \;=\; a_t\!\big(k_i^{(t)};\mathcal{D}\big), \label{eq:va-obs}\\
\textbf{(Optional Update)}\quad  &k_i^{(t+1)} \;=\; \mathcal{U}\!\big(k_i^{(t)},\tau_t,a_t,o_t;\mathcal{D}\big), \label{eq:va-update}
\end{align}
and outputs
\begin{equation}
v^{cla}_i \;=\; \pi_V(k_i^{(t)}),
\end{equation}
where $\pi_V(v,\mathcal{S})=v$ projects the VQL component.

\textit{Tools.}
We model each tool as a predicate-valued mapping that returns a boolean verdict and a diagnostic message $\epsilon$.
Formally, a tool is a mapping
\begin{equation}
\mathcal{A}: \mathcal{V}\times\mathcal{S}\times\mathcal{D}\ \rightarrow\ \{0,1\}\times \mathcal{E},
\end{equation}
where $\mathcal{E}$ denotes the diagnostic space (e.g., success or failure record).

\textit{(1) Syntax Validator.}
Define the syntax predicate mapping: $a_\mathrm{syn}:\mathcal{V}\rightarrow\{0,1\}\times\mathcal{E}:$
\begin{equation}
a_\mathrm{syn}(v)\;=\;\Big(\mathbb{I}[v\in \mathcal{L}(\mathcal{G})],\ \epsilon_{syn}\Big),
\end{equation}
where $\mathbb{I}[\cdot]$ is indicator function and $\mathcal{L}(\mathcal{G})$ denotes the VQL language defined by grammar $\mathcal{G}$. Denote the boolean verdict as
$r^{syn}_t \;=\; \pi_B\!\big(a_\mathrm{syn}(v^{(t)}_i)\big)$,
where $\pi_B(b,\epsilon)=b$ projects the boolean component.

\textit{(2) Schema Validator.}
Define the schema-grounding predicate mapping $a_\mathrm{sch}:\mathcal{V}\times\mathcal{D}\rightarrow \{0,1\}\times\mathcal{E}$:
\begin{equation}
a_\mathrm{sch}(v,\mathcal{D}) \;=\;
\Big(\mathbb{I}[\mathrm{Tabs}(v)\subseteq\mathcal{R}(\mathcal{D})\wedge \mathrm{Cols}(v)\subseteq\mathcal{C}(\mathcal{D})],\ \epsilon_{sch}\Big)
\end{equation}
where $\text{Cols}(\cdot)$ and $\text{Tabs}(\cdot)$ extract the referenced columns and tables, respectively. 
Denote the boolean verdict
$r^{sch}_t \;=\; \pi_B\!\big(a_\mathrm{sch}(v^{(t)}_i,\mathcal{D})\big)$.

\textit{(3) SQL Executor.}
Define the execution predicate mapping $a_\mathrm{exec}:\mathcal{V}\times\mathcal{D}\rightarrow \{0,1\}\times\mathcal{E}$:
\begin{equation}
a_\mathrm{exec}(v,\mathcal{D})\;=\;\Big(\mathbb{I}[\mathrm{Exec}(v,\mathcal{D})],\ \epsilon_{sql}\Big),
\end{equation}
where $\text{Exec}(\cdot)$ denotes the execution of the SQL part of VQL.

\textit{(4) Intent Matcher.}
Define an ambiguity-detection predicate mapping
$a_{\mathrm{int}}:\mathcal{S}\times\mathcal{D}\rightarrow \{0,1\}\times\mathcal{E}$:
\begin{equation}
a_{\mathrm{int}}(\mathcal{S},\mathcal{D})
\;=\;
\Big(\mathbb{I}[\mathrm{Amb}(\mathcal{S},\mathcal{D})],\ \epsilon_{\mathrm{amb}}\Big),
\end{equation}
where $\mathrm{Amb}(\cdot)$ indicates whether the current status $S_i$ is intent-ambiguous under schema $\mathcal{D}$.
If ambiguity is detected, the Validation Agent queries the User Agent for a non-leaking clarification (reference to \ UA-Rule~2):
\begin{equation}
\mathbb{I}[\mathrm{Amb}(\mathcal{S}_i,\mathcal{D})]=1
\ \Rightarrow\
q_i^{clr}=\mathrm{UserAgent}(\phi_i)\in\mathcal{Q}^{clr}\cup\{\bot\},
\end{equation}
where $\phi_i$ denotes clarification request issued by Validation Agent.

\textit{Tool permission predicate.}
Define a state-dependent permission predicate
$\delta:\mathcal{K}\times\mathcal{A}\rightarrow\{0,1\}$,
where $\delta(k_i^{(t)},a)=1$ indicates that action/tool $a$ is permitted to be selected at state $k_i^{(t)}$.

\textit{Validation Agent (VA) rules.}
We constrain the admissible reasoning--acting process by the following rules.

\noindent\textbf{VA-Rule 1 (Bounded Loop):}
$\exists T\le m,\quad v^{cla}_i=v^{(T)}_i.$

\noindent\textit{Explanation.} The Validation Agent must terminate after at most $m$ ReAct steps.

\noindent\textbf{VA-Rule 2 (Syntax-first): }
$\forall t, r^{syn}_t=0 \Rightarrow\
\delta(c_t,a_{\mathrm{sch}}) \land\
\delta(c_t,a_{\mathrm{exec}}) \land\
\delta(c_t,a_{\mathrm{int}})=0.$

\noindent\textit{Explanation.} If syntax is invalid, the agent is forbidden to select schema/execution/intent actions at this step, and must prioritize syntax repair.

\noindent\textbf{VA-Rule 3 (Schema-before-Execution):}

\noindent$\forall t,\quad r^{sch}_t=0\ \Rightarrow\ \delta(c_t,a_{\mathrm{exec}})=0.$

\noindent\textit{Explanation.} SQL execution is only permitted after schema grounding succeeds.

\noindent\textbf{VA-Rule 4 (Observation-guided Update).}
Define a repair mapping
\begin{equation}
\mathcal{U}:\mathcal{K}\times\mathcal{T}\times\mathcal{A}\times\mathcal{O}\times\mathcal{D}\ \rightarrow\ \mathcal{V},
\end{equation}
which produces the next candidate by conditioning on the current candidate, dialogue status, and the tool feedback:
\begin{equation}
\forall t,\quad 
v^{(t+1)}_i \;=\; \mathcal{U}\!\big(k^{(t)}_i,\ \tau_t,\ a_t,\ o_t;\ \mathcal{D}\big).
\end{equation}
\textit{Explanation.} The tool used and feedback is incorporated to guide the repair that yields the next candidate $v^{(t+1)}_i$.

\begin{algorithm}[t]
\caption{PMVisAgent Algorithm}
\label{alg:pmvisagent}
\KwIn{NLQ sequence $\mathcal{Q}=\{q_1,\ldots,q_n\}$; database schema $\mathcal{D}$; tool set $\mathcal{A}$; max iterations $m$}
\KwOut{Validated VQL sequence $\mathcal{V}^{cla}=\{v^{cla}_1,\ldots,v^{cla}_n\}$}

\BlankLine
\textbf{Initialize:} $\mathcal{H}_0 \gets \emptyset$; $\mathcal{V}^{cla} \gets [\,]$\;

\SetKwProg{Pr}{Procedure}{:}{}
\Pr{\textsc{PMVisAgent$(\mathcal{Q}, \mathcal{D}, \mathcal{A}, m$)}}{
\For{$i \gets 1$ \KwTo $n$}{
  \tcp{\color{darkblue} User Agent issues NLQ}
  $q_i \gets \mathcal{Q}[i]$\;
  $\mathcal{S}_i \gets \mathcal{H}_{i-1} \cup \{q_i\}$\;

  \tcp{\color{darkblue} System Agent translates NLQ to candidate VQL}
  $p \gets \textsc{Prompt-Maker}(\mathcal{S}_i,\mathcal{D})$\;
  
  $v^{gen}_i \gets \textsc{Call-LLM}(p)$
  
  \tcp{\color{darkblue} Validation Agent validates and refines}
  $v^{cla}_i \gets \textsc{Validation-Agent}(q_i,v^{gen}_i,\mathcal{D},\mathcal{A},m)$\;

  $\mathcal{V}^{cla}$.append($v^{cla}_i$)\;
  $\mathcal{H}_i \gets \mathcal{H}_{i-1} \cup \{(q_i, v^{cla}_i)\}$\;
}
\Return $\mathcal{V}^{cla}$\;
}
\BlankLine
\SetKwProg{Fn}{Function}{:}{}

\Fn{\textsc{Validation-Agent}$(q,v,\mathcal{D},\mathcal{A},m)$}{
  $\mathcal{O}\gets[\,]$\;
  \For{$t \gets 1$ \KwTo $m$}{
    $\mathcal{R} \gets \textsc{Call-LLM}(q, v, \mathcal{D}, \mathcal{O})$\;
    $v_i^{cla}$ = \textsc{Extract-Final-Answer}($\mathcal{R}$)\;
    \uIf{$v_i^{cla}$}{
        \Return $v_i^{cla}$
    }
    \Else{
        $\mathcal{A} \gets \textsc{Extract-Action}(\mathcal{R})$\;
        $\mathcal{O} \gets \textsc{Call-Tool}(\mathcal{A}, q, v,\mathcal{D})$\;
        $v_{upd} \gets \textsc{Extract-Update-VQL}(\mathcal{R})$\;
        \uIf{$v_{upd}$}{
            $v \gets v_{upd}$
        }
    }
  }
  \Return $v$
}
\label{algoritm:PMVisAgent}

\end{algorithm}

Algorithm \ref{algoritm:PMVisAgent} outlines the overall workflow of PMVisAgent. For each NLQ issued by User Agent, the System Agent first generates a candidate VQL through prompt construction and LLM inference. The Validation Agent then iteratively calls LLM to validate the VQL. If a final answer is detected in the response, it is returned directly; otherwise, a detected tool is executed, and the candidate VQL is refined until validation is complete or the iteration limit is reached. 

\subsection{Formal Guarantees}
Our PMVisAgent framework provides mathematical guarantees through the following fundamental theorems.

\begin{theorem}[Tool-action Precedence Invariant]
\label{thm:precedence}
Along any validation trace in any round $i$, the following precedence constraints hold:
\begin{align}
a_t \in \{a_{sch},a_{exec},a_{int}\} \ &\Rightarrow\ r_t^{syn}=1, \label{eq:precedence_syn}\\
a_t = a_{exec} \ &\Rightarrow\ r_t^{sch}=1. \label{eq:precedence_sch}
\end{align}
Equivalently, schema validation / SQL execution / intent matching can only be taken when the current candidate
is syntactically valid; and SQL execution can only be taken when schema grounding succeeds.
\end{theorem}
\begin{proof}[Proof sketch]
We prove this by contradiction and the permission predicate $\delta$.
If $a_t \in \{a_{sch},a_{exec},a_{int}\}$ is selected, then necessarily $\delta(k_i^{(t)},a_t)\\=1$.
Assume $r_t^{syn}=0$. By VA-Rule~2, we have
$\delta(k_i^{(t)},a_{sch})=\delta(k_i^{(t)},a_{exec})=\delta(k_i^{(t)},a_{int})=0$,
contradicting $\delta(k_i^{(t)},a_t)=1$. Hence $r_t^{syn}=1$, proving Eq.~(\ref{eq:precedence_syn}).
Similarly, if $a_t=a_{exec}$ and assume $r_t^{sch}=0$, VA-Rule~3 gives
$\delta(k_i^{(t)},a_{exec})=0$, contradicting $\delta(k_i^{(t)},a_{exec})=1$.
Hence $r_t^{sch}=1$, proving Eq.~(\ref{eq:precedence_sch}).
\end{proof}

\begin{corollary}[Safe SQL Execution]
\label{cor:safe_exec}
Whenever the Validation Agent executes SQL at step $t$ (i.e., $a_t=a_{exec}$), the executed candidate $v_i^{(t)}$ is
(i) syntactically valid and (ii) schema-grounded:
\begin{align}
a_t=a_{exec} \Rightarrow\ v_i^{(t)}\in \mathcal{L}(\mathcal{G}) \wedge\ Tabs(v_i^{(t)})\subseteq \mathcal{R}(\mathcal{D}) \nonumber\\
\wedge\ Cols(v_i^{(t)})\subseteq \mathcal{C}(\mathcal{D}). \nonumber
\end{align}
\end{corollary}
\begin{proof}[Proof sketch]
By Theorem~\ref{thm:precedence}, $a_t=a_{exec}$ implies $r_t^{syn}=1$ and $r_t^{sch}=1$.
By the tool semantics, $r_t^{syn}=1$ indicates $v_i^{(t)}\in \mathcal{L}(\mathcal{G})$ (syntax validator),
and $r_t^{sch}=1$ indicates $Tabs(v_i^{(t)})\subseteq \mathcal{R}(\mathcal{D})$ and
$Cols(v_i^{(t)})\subseteq \mathcal{C}(\mathcal{D})$ (schema validator).
Combining them yields the claim.
\end{proof}

\begin{theorem}[Bounded Tool-call Complexity]
\label{thm:tool_bound}
Across an $n$-round dialogue, the total number of tool invocations issued by the Validation Agent is upper-bounded by $nm$.
\end{theorem}
\begin{proof}[Proof sketch]
Fix a round $i$. By VA-Rule~1, the validation trace length satisfies $T\le m$.
From the control flow of Algorithm~1, each ReAct step executes a single selected action $a_t$,
hence each step triggers at most one external tool invocation.
Let $\mathbb{I}[t]\in\{0,1\}$ indicate whether step $t$ invokes a tool.
Then the number of tool invocations in round $i$ is
\begin{equation}
N_i \triangleq \sum_{t=1}^{T}\mathbb{I}[t] \ \le\ \sum_{t=1}^{T}1 \ =\ T \ \le\ m.
\end{equation}
Summing over $n$ rounds yields $\sum_{i=1}^{n}N_i \le \sum_{i=1}^{n} m = nm$.
\end{proof}

\begin{corollary}[Time Bound of Validation]
\label{cor:time_bound}
Let $x_{\textsc{llm}}$ be an upper-bound on the time of producing $(\tau_t,a_t)$ in one ReAct step,
and let $c_{\max}$ be an upper-bound on the time of one tool invocation.
Let $N_i$ be the number of tool invocations issued in round $i$.
Then the worst-case validation time satisfies
\begin{align}
\textsc{Time}(i) &\le m\cdot x_{\textsc{llm}} + N_i\cdot x_{\max}, \label{eq:time_decomp}\\
\sum_{i=1}^{n}\textsc{Time}(i) &\le nm\cdot x_{\textsc{llm}} + nm\cdot x_{\max}
= nm\big(x_{\textsc{llm}}+x_{\max}\big). \label{eq:time_total}
\end{align}
\end{corollary}
\begin{proof}[Proof sketch]
In round $i$, there are at most $T\le m$ steps (VA-Rule~1), and each step costs at most $x_{\textsc{llm}}$,
thus the total reasoning time is at most $m\cdot x_{\textsc{llm}}$.
Moreover, each tool invocation costs at most $x_{\max}$, so tool time is at most $N_i\cdot x_{\max}$,
which proves Eq.~(\ref{eq:time_decomp}).
Summing Eq.~(\ref{eq:time_decomp}) over $i=1,\dots,n$ gives
\[
\sum_{i=1}^{n}\textsc{Time}(i)\ \le\ nm\cdot x_{\textsc{llm}} + \Big(\sum_{i=1}^{n}N_i\Big)\cdot x_{\max}.
\]
By Theorem~\ref{thm:tool_bound}, $\sum_{i=1}^{n}N_i \le nm$, yielding Eq.~(\ref{eq:time_total}).
\end{proof}

\begin{theorem}[System Reliability]
\label{thm:reliability}
For any $n$-round dialogue execution context $\Gamma$, PMVisAgent satisfies overall system reliability:
\begin{equation}
\mathcal{Y}(\Gamma)
\;\triangleq\;
\Theta_{\text{safety}}(\Gamma) \ \wedge\ \Theta_{\text{tool}}(\Gamma) \ \wedge\ \Theta_{\text{time}}(\Gamma),
\label{eq:reliability}
\end{equation}
where $\Theta_{\text{safety}}(\Gamma)$ denotes execution safety, $\Theta_{\text{tool}}(\Gamma)$ denotes bounded tool-call complexity, and $\Theta_{\text{time}}(\Gamma)$ denotes bounded validation time.
\end{theorem}
\begin{proof}[Proof sketch]
Define system reliability as the logical conjunction in Eq.~(\ref{eq:reliability}). It satisfies: 
\textbf{(1) Safety.}
$\Theta_{\text{safety}}(\Gamma)$ requires that any SQL execution in any round/step is performed only on a syntactically valid and schema-grounded
candidate. This follows from the tool-action precedence invariant (Theorem~\ref{thm:precedence})
and Safe SQL Execution (Corollary~\ref{cor:safe_exec}), which together quantify over all execution steps.
\textbf{(2) Bounded tool calls.}
$\Theta_{\text{tool}}(\Gamma)$ states that the total number of tool invocations across $n$ rounds is upper-bounded by $nm$.
This is exactly Theorem~\ref{thm:tool_bound}, derived from VA-Rule~1 ($T\le m$) and Algorithm~1's single-action control flow.
\textbf{(3) Bounded time.}
$\Theta_{\text{time}}(\Gamma)$ states that the overall runtime is bounded by the sum of (i) per-step reasoning time and
(ii) per-tool invocation time, with the number of invocations bounded by $\Theta_{\text{tool}}(\Gamma)$.
This is established in Corollary~\ref{cor:time_bound}, which uses Theorem~\ref{thm:tool_bound} to bound $\sum_i N_i$.

Since each component holds for any execution trace under $\Gamma$, their conjunction yields
$\mathcal{Y}(\Gamma)$.
\end{proof}

\section{Experiments}

\subsection{Experimental Setup}

\subsubsection{Dataset}

All experiments were conducted on PMVisBench, the benchmark dataset introduced in Section~\ref{sec:PMVisBench}. PMVisBench is specifically designed for evaluating text-to-vis systems based on the PMVis paradigm and covers a wide range of progressive multi-turn refinement scenarios. 
Each sample consists of a sequence of NLQs and their corresponding expected VQLs, emphasizing the progressive nature of real-world user-system interactions.

\begin{table*}
\centering
\caption{Performance comparison of PMVisAgent against baseline models on PMVisBench, where $^*$ represents traditional neural network-based models and $^\dagger$ represents the LLM-based models. Results are reported separately for single-table (w/o join operations) and multi-table (w/ join operations) settings. The best results of each metric are highlighted in \textcolor{red}{\textbf{bold red}}, while the second-best results are in \textcolor{black}{\textbf{bold black}}.}
\resizebox{\linewidth}{!}{
\begin{tabular}{l|ccccc|ccccc}
\toprule
\multicolumn{1}{c|}{\multirow{2}{*}{Model}} & \multicolumn{5}{c|}{Single-Table (w/o join operation) (\%)} & \multicolumn{5}{c}{Multi-Table (w/ join operation) (\%)} \\ 
\multicolumn{1}{c|}{}                       & Vis ACC.   & Axis ACC.   & Data ACC.  & ACC.   & Exec. ACC.  & Vis ACC.   & Axis ACC.   & Data ACC.   & ACC.   & Exec. ACC.  \\ \midrule
Seq2Vis$^*$                                     & 75.13      & 1.39        & 0.44       & 0.00   & 0.38       & 73.90      & 2.46        & 0.64        & 0.00   & 0.96       \\
Transformer$^*$                                  & 92.82      & 14.55       & 14.55      & 6.61   & 9.95       & 91.87      & 19.25       & 14.65       & 9.30   & 14.22      \\
ncNet$^*$                                        & 97.04      & 17.00       & 18.01      & 11.08  & 13.92      & \textcolor{red}{\textbf{99.89}}      & 18.40       & 14.97       & 12.09  & 16.79      \\
RGVisNet$^*$                                     & 97.04          & 64.48           & 48.99          & 43.01      & 67.51          & 97.00           & \textbf{26.20}           & 18.83           & 15.51      & 49.63          \\
MMCoVisNet$^*$                                   & 97.86          & 37.85           & 26.69           & 23.68      & 44.52          & 97.75          & 20.00           & 12.78           & 13.05      & 31.98          \\
nvAgent$^\dagger$                                  & 93.07          & 51.20           & 49.94           & 36.46      & 67.13          & 95.83          & 0.00           & 0.53           & 0.00      & 52.83          \\
Prompt4Vis$^\dagger$                                  & 98.11          & 69.96           & 57.87           & 50.31      & 71.03          & 94.33          & \textcolor{red}{\textbf{31.87}}           & \textbf{28.98}           & \textcolor{red}{\textbf{20.43}}      & 54.65          \\
PMVisAgent (gpt-4o-mini)$^\dagger$                     & \textcolor{red}{\textbf{99.31}}      & \textbf{72.54}       & \textbf{75.44}      & \textbf{68.07}  & \textbf{83.31}      & \textbf{99.36}      & 6.52        & 18.82       & 5.13   & \textbf{65.35}      \\
PMVisAgent (gemini)$^\dagger$                         & 89.29      & 63.16       & 62.66      & 59.38  & 76.39      & 76.47      & 4.39        & 12.62       & 3.10      & 53.48          \\
PMVisAgent (qwen-plus)$^\dagger$                      & \textbf{98.80}          & \textcolor{red}{\textbf{86.96}}           & \textcolor{red}{\textbf{81.55}}          & \textcolor{red}{\textbf{79.85}}      & \textcolor{red}{\textbf{88.60}}          & 98.18          & 22.78           & \textcolor{red}{\textbf{32.94}}           & \textbf{17.01}      & \textcolor{red}{\textbf{77.86}}          \\ \bottomrule
\end{tabular}}
\label{tab:main result}
\end{table*}

\begin{table*}[!t]
\centering
\caption{Ablation results of PMVisAgent. This table shows the performance of PMVisAgent after removing the key components.}
\resizebox{\linewidth}{!}{
\begin{tabular}{l|ccccc|ccccc}
\toprule
\multicolumn{1}{c|}{\multirow{2}{*}{Model}} & \multicolumn{5}{c|}{Single Table (w/o join operation) (\%)}                        & \multicolumn{5}{c}{Multiple Tables (w/ join operation) (\%)}                        \\ 
\multicolumn{1}{c|}{}                       & Vis ACC.       & Axis ACC.      & Data ACC.      & ACC.           & Exec. ACC.      & Vis ACC.        & Axis ACC.      & Data ACC.      & ACC.           & Exec. ACC.      \\ \midrule
PMVisAgent (gpt4o-mini)                     & \textbf{99.31} & \textbf{72.54} & \textbf{75.44} & \textbf{68.07} & \textbf{83.31} & 99.36           & \textbf{6.52}  & \textbf{18.82} & \textbf{5.13}  & \textbf{65.35} \\
- Only Single-turn                          & 98.11          & 44.14          & 51.83          & 32.43          & 63.60          & 99.25           & 1.60           & 6.42           & 0.75           & 46.62          \\
- Only Multi-turn                           & 98.99          & 47.48          & 56.68          & 38.92          & 68.20          & \textbf{99.57}  & 1.39           & 5.99           & 0.64           & 51.12          \\
- Single-turn + Validation
& 98.17          & 64.11          & 71.22          & 59.51          & 78.97          & 98.40           & 4.49           & 16.36          & 2.89           & 63.42          \\ \midrule
PMVisAgent (gemini)                         & 89.29          & \textbf{63.16} & \textbf{62.66} & \textbf{59.38} & \textbf{76.39} & 76.47           & \textbf{4.39}  & \textbf{12.62} & \textbf{3.10}  & \textbf{53.48} \\
- Only Single-turn                          & \textbf{99.87} & 36.27          & 38.54          & 28.21          & 60.01          & \textbf{100.00} & 2.14           & 9.30           & 1.50           & 47.91          \\
- Only Multi-turn                           & 99.37          & 34.38          & 39.55          & 28.15          & 59.70          & 99.47           & 1.71           & 7.38           & 0.96           & 51.87          \\
- Single-turn + Validation                     & 87.47          & 51.32          & 54.22          & 48.36          & 72.61          & 74.12           & 3.64           & 11.02          & 2.99           & 50.91          \\ \midrule
PMVisAgent (qwen-plus)                      & 98.80          & \textbf{86.96} & \textbf{81.55} & \textbf{79.85} & \textbf{88.60} & 98.18           & \textbf{22.78} & \textbf{32.94} & \textbf{17.01} & \textbf{77.86} \\
- Only Single-turn                          & 99.81          & 51.70          & 47.61          & 32.62          & 64.99          & 99.79           & 4.18           & 11.35          & 0.96           & 49.79          \\
- Only Multi-turn                           & \textbf{99.94} & 54.22          & 60.20          & 45.34          & 72.67          & \textbf{99.89}  & 4.07           & 9.53           & 1.50           & 50.86          \\
- Single-turn + Validation                     & 99.12          & 78.02          & 75.88          & 70.28          & 85.20          & 97.22           & 13.92          & 23.77          & 8.14           & 66.17          \\ \bottomrule
\end{tabular}}
\label{tab:ablation}
\end{table*}

\begin{table*}[t]
\small
\centering
\caption{Case Study. VQLs generated by different baselines and PMVisAgent on an NLQ from PMVisBench, together with their corresponding visualization charts. The correct chart is marked with \ding{52}, while the error charts are marked with \ding{56}.}
\vspace{-5pt}
\begin{minipage}[t]{0.71\textwidth}
\vspace{0.5cm} 
\setlength{\tabcolsep}{4pt}
\renewcommand{\arraystretch}{1.1}
\begin{tabular}{>{\centering\arraybackslash}m{2cm}|m{0.78\linewidth}}
\toprule
NLQ         & Please compare the number of floors for all street address using a bar chart, and order by the number of floors in ascending.\\
\midrule
Target VQL  & Visualize BAR SELECT Street\_address , Floors FROM building ORDER BY floors $\to$ \textit{Figure (a) in Table III} \\
\midrule
Seq2Vis     & Visualize BAR SELECT name , COUNT(*) FROM products AS T1 JOIN ORDER BY floors $\to$ \textit{Syntax Error}\\
\midrule
Transformer & Visualize BAR SELECT Type , COUNT(*) FROM institution GROUP BY TYPE ORDER BY Type ASC $\to$ \textit{Figure (b) in Table III} \\
\midrule
ncNet       & Visualize BAR SELECT type , COUNT(type) FROM type ORDER BY type ASC $\to$ \textit{no such table: type} \\
\midrule
RGVisNet    & Visualize BAR SELECT Street\_address, Floors FROM building GROUP BY Street\_address ORDER BY COUNT(Floors) ASC $\to$ \textit{Figure (c) in Table III} \\
\midrule
MMCoVisNet & Visualize BAR SELECT name, COUNT(*) FROM institution ORDER BY COUNT(*) DASC $\to$ \textit{no such column: name} \\
\midrule
\multirow{2}{*}[0.75em]{\makecell{Prompt4Vis \\ \&  nvAgent}} & Visualize BAR SELECT Street\_address, COUNT(Floors) FROM building GROUP BY Street\_address ORDER BY COUNT(Floors) ASC $\to$ \textit{Figure (d) in Table III} \\
\midrule
PMVisAgent  & Visualize BAR SELECT Street\_address, Floors FROM building ORDER BY Floors $\to$ \textit{Figure (a) in Table III} \\
\bottomrule
\end{tabular}
\end{minipage}
\hfill
\begin{minipage}[t]{0.28\textwidth}
\vspace{0pt}
\centering
\begin{subcaptionblock}{0.38\linewidth}
  \includegraphics[width=\linewidth]{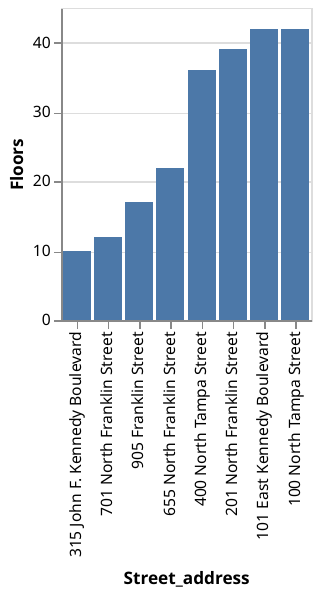}
  \caption{\textcolor{black}{\ding{52}}}
  \label{fig:case_gt}
\end{subcaptionblock}\hfill
\begin{subcaptionblock}{0.25\linewidth}
  \includegraphics[width=\linewidth]{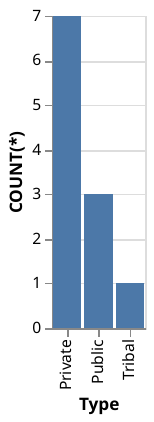}
  \caption{\textcolor{black}{\ding{56}}}
  \label{fig:case_transformer}
\end{subcaptionblock}

\medskip
\vspace{-0.2cm}
\begin{subcaptionblock}{0.43\linewidth}
  \includegraphics[width=\linewidth]{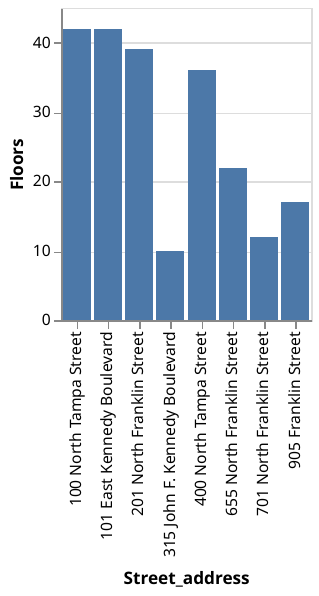}
  \caption{\textcolor{black}{\ding{56}}}
  \label{fig:case_rgvisnet}
\end{subcaptionblock}\hfill
\begin{subcaptionblock}{0.43\linewidth}
  \includegraphics[width=\linewidth]{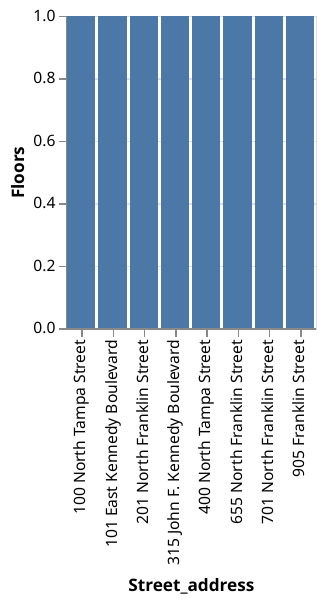}
  \caption{\textcolor{black}{\ding{56}}}
  \label{fig:case_prompt4vis}
\end{subcaptionblock}

\end{minipage}
\label{tab:case_study}
\vspace{-10pt}
\end{table*}

\subsubsection{Baselines}

To evaluate the effectiveness of PMVisAgent, we compare it with representative text-to-vis baselines, covering both traditional neural models and recent LLM-based methods:

\begin{itemize}[leftmargin=*]
    \item \textbf{Seq2Vis}~\cite{luo2021synthesizing}: A seq2seq baseline that directly translates NLQs into visualization specifications.

    \item \textbf{Transformer}~\cite{vaswani2017attention}: A standard Transformer encoder--decoder that generates visualization language with self-attention.

    \item \textbf{ncNet}~\cite{luo2021natural}: A Transformer-based text-to-vis model with visualization-aware optimizations (e.g., attention-forcing and rendering) and optional chart templates.

    \item \textbf{RGVisNet}~\cite{song2022rgvisnet}: A retrieval--generation framework that retrieves similar queries and refines them for the input NLQ and schema.

    \item \textbf{MMCoVisNet}~\cite{song2024marrying}: A multi-modal conversational network for CoVis, using a data-aware dialogue encoder with adaptive decoders for different response types.

    \item \textbf{Prompt4Vis}~\cite{li2025prompt4vis}: A state-of-the-art LLM-based framework with example mining and schema filtering for in-context VQL generation.
    
    \item \textbf{nvAgent}~\cite{ouyang-etal-2025-nvagent}: A collaborative agent workflow, including processor, composer, and validator for processing, planning, and verification, but with a fixed workflow.

    \item \textbf{PMVisAgent}: Our progressive multi-turn agent framework for text-to-vis. It incrementally refines VQLs across dialogue rounds and adopts a ReAct-style tool-use loop with explicit interaction and verification rules to validate and repair generated VQLs, mitigating error accumulation and improving execution correctness.
\end{itemize}

For transparency and fair comparison, all competing models are implemented following the experimental settings reported in their original papers whenever applicable.

\subsubsection{Evaluation Metrics}

Following prior work~\cite{luo2021synthesizing,song2022rgvisnet,li2025prompt4vis}, we report four widely adopted metrics: \textit{Vis Accuracy}, \textit{Data Accuracy}, \textit{Axis Accuracy}, and \textit{Overall Accuracy}. Since semantically equivalent VQLs may differ syntactically, we additionally report \textit{Execution Accuracy}, which checks whether the predicted VQL yields the same execution result as the target on the underlying database, preventing underestimating of model performance.

Let $N$ be the total number of samples in the test set. Let $N_{Vis}$, $N_{Axis}$, $N_{Data}$, and $N_o$ denote the numbers of predictions that are exactly matched under the corresponding criteria, and let $N_{Exec}$ be the number of predictions that match the execution results.

\begin{itemize}[leftmargin=*, itemsep=2pt, topsep=2pt, parsep=0pt, partopsep=0pt]
    \item \textbf{Vis Accuracy}: chart type match, $\text{Vis Acc.}=\frac{N_{Vis}}{N}$.
    \item \textbf{Axis Accuracy}: axis attribute match, $\text{Axis Acc.}=\frac{N_{Axis}}{N}$.
    \item \textbf{Data Accuracy}: data transformations match, $\text{Data Acc.}=\frac{N_{Data}}{N}$.
    \item \textbf{Overall Accuracy}: exact match on chart type, axes and data transformations, $\text{Acc.}=\frac{N_o}{N}$.
    \item \textbf{Execution Accuracy}: the execution results of the method match the execution results of the target, $\text{Exec Acc.}=\frac{N_{Exec}}{N}$.
\end{itemize}

\subsubsection{Implementation Details}

For the neural network-based baselines that require training, we conduct five-fold cross-validation on PMVisBench. In each fold, we split the data into train/validation/test sets with a 60\%/20\%/20\% ratio. We report metrics on the union of the five test folds to cover the full benchmark.

PMVisAgent requires no task-specific training and is evaluated in a zero-shot setting with different LLMs. Specifically, we use \xwx{\texttt{gpt-4o-mini} (2024-07-18, 16k max output), \texttt{gemini-2.5-flash-lite} (2025-07-22, 16k max output), and \texttt{qwen-plus} (2025-01-25, 8k max output). We set temperature, frequency\_penalty, and presence\_penalty to 0, and set the maximum ReAct steps for the Validation Agent ($m$) to 10.}

\subsection{Performance Comparison}

To comprehensively evaluate the effectiveness of PMVisAgent, we compare it with representative text-to-vis baselines across both single-table (without join operation) and multi-table (with join operation) settings. The main results are summarized in Tab.~\ref{tab:main result}.

Overall, traditional neural network models (Seq2Vis, Transformer, ncNet, RGVisNet, MMCoVisNet) show clear limitations in text-to-vis tasks, especially under execution-based evaluation. Seq2Vis achieves only 0.38\% and 0.96\% execution accuracy in single-table and multi-table settings, respectively. Transformer and ncNet improve to 9.95\%/14.22\% and 13.92\%/16.79\%. RGVisNet performs better with retrieval augmentation (67.51\%/49.63\%), while MMCoVisNet remains limited (44.52\%/31.98\%), despite being designed for multi-round interaction.
Regarding LLM-based models, nvAgent outperforms almost all traditional neural network-based baselines, with execution accuracy reaching 67.13\% and 52.83\% in single-table and multi-table settings, respectively. The strongest LLM-based baseline, Prompt4Vis, achieves a higher 71.03\% execution accuracy on single-table, but drops sharply to 54.65\% in multi-table cases, showing the model struggles to maintain effectiveness when more complex join operations are required.

In contrast, PMVisAgent consistently outperforms all baselines. With qwen-plus, it achieves the best execution accuracy of 88.60\% on single-table (+17.57\% vs. Prompt4Vis) and 77.86\% on multi-table (+23.21\% vs. Prompt4Vis). It may score slightly lower on some clause-level metrics in multi-table cases due to strict exact-match requirements, while its higher execution accuracy indicates stronger semantic correctness and robustness to syntactic variations.

Across LLM backbones, PMVisAgent remains strong, where gpt-4o-mini reaches 83.31\%/65.35\% (single-/multi-table), beating Prompt4Vis by +12.28\%/+10.7\%; gemini improves single-table by +5.36\% and is comparable on multi-table (-1.17\%). Overall, PMVisAgent generalizes well and is largely plug-and-play.

\subsection{Ablation Study}\label{sec:ablation study}

In this subsection, we conduct ablation studies using the same LLM backbones to distinguish the contribution of system design from that of the underlying LLMs. We quantify the impact of progressive multi-turn interaction and validation in PMVisAgent through three variants:
\begin{enumerate}[leftmargin=*]
    \item \textit{Only Single-turn}: the System Agent generates VQLs in a single round, without progressive multi-turn refinement and validation (i.e., the raw LLM output);
    \item \textit{Only Multi-turn}: progressive multi-turn refinement is preserved but Validation agent is disabled;
    \item \textit{Single-turn + Validation}: Validation Agent is preserved but interaction is restricted to the one-shot paradigm.
\end{enumerate}
The full PMVisAgent integrates both progressive multi-turn interaction and validation, serving as the reference configuration.

The results of ablation experiments are shown in Tab.~\ref{tab:ablation}. Removing both multi-turn refinement and validation (\textit{Only Single-turn}) in PMVisAgent causes a substantial drop across metrics. For example, in the single-table setting, execution accuracy decreases by 19.71\% (gpt-4o-mini), 16.38\% (gemini), and 23.61\% (qwen-plus), and similar trends hold for the multi-table setting. Since these variants use the same LLMs as the full PMVisAgent, the gains cannot be attributed merely to stronger model backbones; instead, they reflect the effect of the proposed interaction and validation design. Moreover, disabling validation (\textit{Only Multi-turn}) also results in accuracy degradation since errors accumulate across rounds, reducing single-table execution accuracy by 15.11\%, 16.69\%, and 15.93\% for the three LLMs, respectively. Keeping validation but removing progressive interaction (\textit{Single-turn + Validation}) improves over \textit{Only Single-turn} (e.g., qwen-plus execution accuracy 85.20\% vs.\ 64.99\%), yet falls short of the full PMVisAgent framework (88.60\%), indicating that validation alone cannot compensate for the lack of progressive user refinement, and that both components are complementary.

In summary, combining progressive multi-turn interaction with validation yields the best and most stable performance, supporting PMVis by reflecting real user behavior while ensuring robustness against error accumulation.

\begin{table}[!t]
\centering
\caption{Cost analysis. Comparing PMVisAgent with agent-based baseline nvAgent on the average number of interaction rounds, token usage and latency per round, together with execution accuracy under the same LLM backbone settings.}
\resizebox{\linewidth}{!}{
\begin{tabular}{l|c|c|c|c}
\hline
\multicolumn{1}{c|}{Model} & \# of round & \makecell[c]{Tokens\\/round}                & \makecell[c]{Latency\\(s/round)}              & Exec. ACC.           \\ \hline
nvAgent (gpt-4o-mini)      & 1           & 5349.1                & 6.97                  & 61.83                     \\
PMVisAgent (gpt-4o-mini)   & 3.7         & 1907.6                & 4.56                  & 76.65                \\
nvAgent (gemini)           & 1           & 6733.8                & 8.64                  & 64.61 \\
PMVisAgent (gemini)        & 3.7         & 1983.4                & 6.69                  & 67.90                \\
nvAgent (qwen-plus)        & 1           & 5999.5                & 21.38                 & 70.99 \\
PMVisAgent (qwen-plus)     & 3.7         & 2073.4                & 10.04                  & 84.62                \\ \hline
\end{tabular}}
\label{tab:cost}
\end{table}

\begin{figure*}[ht!]
    \centering
    \includegraphics[width=\textwidth]{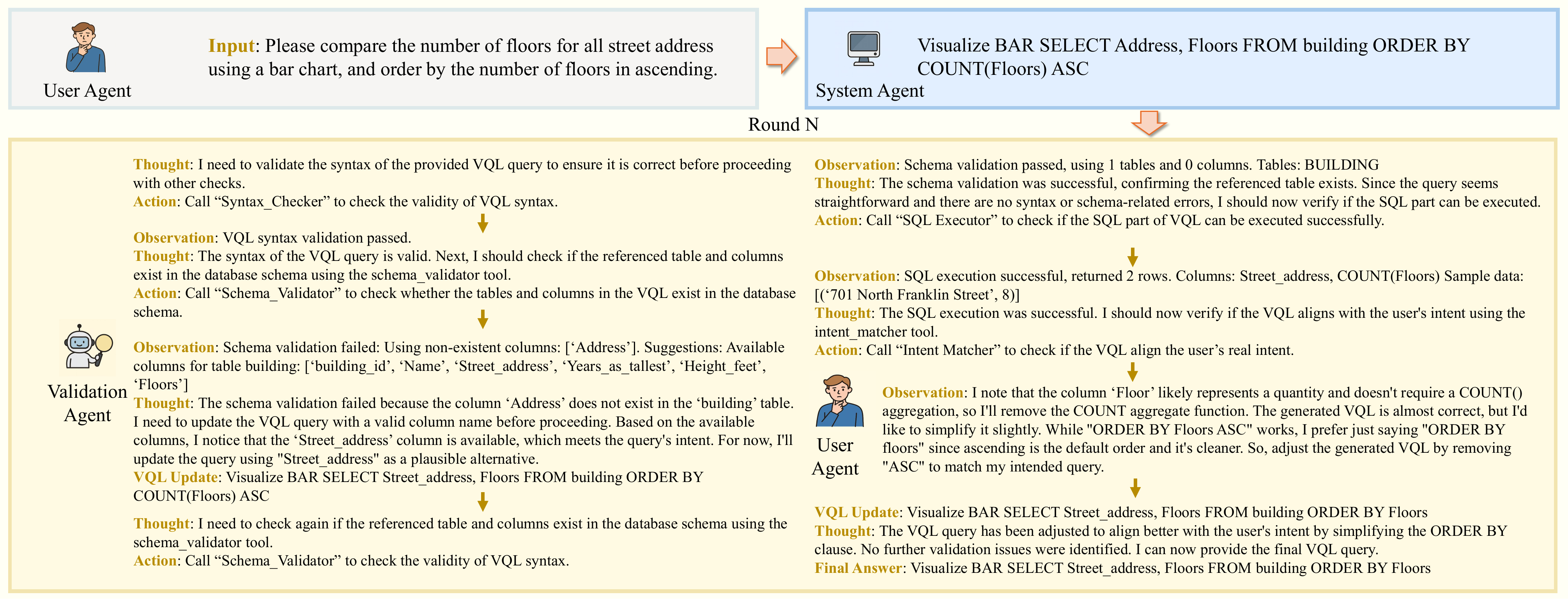}
    \caption{An example of PMVisAgent workflow. This case shows how Validation Agent detects and repairs both schema errors (e.g., replacing a non-existent field) and intent mismatches (e.g., removing unnecessary aggregation and simplifying ordering).}
    \label{fig:case_study}
\end{figure*}

\subsection{Cost and Efficiency Analysis}

\xwx{To explore the runtime and computational cost introduced by PMVisAgent, we compare it against nvAgent, a representative agent-based baseline. We report the average number of rounds, token usage and latency per round, and additionally compare execution accuracy under the same LLM backbone settings, as shown in Table \ref{tab:cost}.}

\xwx{The traditional one-shot paradigm typically supports only a single iteration, while the proposed PMVis paradigm enables multi-round refinement, where each sample requires 3.7 rounds on average in PMVisBench. Notably, although PMVisAgent performs more rounds, its overhead of each round is substantially lower than that of nvAgent. Specifically, PMVisAgent consumes only about 1.9k tokens per round, compared to 5.3k tokens per round for nvAgent, and achieves a lower latency. This indicates that the end-to-end overhead remains within an acceptable range and is of the same order of magnitude as nvAgent in practice. More importantly, at a similar cost, PMVisAgent improves over nvAgent by 3.29\%–14.82\% in execution accuracy. Overall, PMVis delivers significant performance improvements without a noticeable increase in computational overhead, demonstrating a favorable efficiency–effectiveness trade-off.}

\begin{table}[!t]
\centering
\caption{Results of user study. The table shows the performance of PMVisAgent under different configurations.}
\resizebox{\linewidth}{!}{
\begin{tabular}{l|ccccc}
\toprule
\multicolumn{1}{c|}{Model} & Vis ACC.        & Axis ACC.      & Data ACC.      & ACC.           & Exec. ACC.     \\ \midrule
PMVisAgent (qwen-plus)     & 97.50           & \textbf{51.00} & \textbf{57.50} & \textbf{47.50} & \textbf{79.50} \\
- Only Single-turn         & \textbf{100.00} & 27.00          & 38.00          & 21.00          & 62.50          \\
- Only Multi-turn          & \textbf{100.00} & 33.50          & 39.50          & 28.50          & 70.50          \\
- Single-turn + Validation & 99.50           & 48.50          & 55.50          & 43.50          & 75.50          \\ \bottomrule
\end{tabular}}
\label{tab:user study}
\end{table}

\subsection{User Study}

To evaluate PMVis from a user-centered perspective, we conducted a user study with 20 student participants to examine how non-expert users express and refine visualization intents through progressive interaction. We randomly sampled 200 cases from PMVisBench. Each participant completed 10 cases by issuing progressive NLQs to reach their intended visualizations, with iterative refinement encouraged to mimic exploratory analysis. We then executed all 200 query trajectories with PMVisAgent (qwen-plus) under the same settings as the ablation study (Section~\ref{sec:ablation study}).

Results are shown in Tab.~\ref{tab:user study}. PMVisAgent achieves 79.5\% execution accuracy, outperforming the single-turn setting (62.5\%). Enabling only multi-turn refinement or only validation improves accuracy to 70.5\% and 75.5\%, respectively, indicating that progressive refinement and verification both contribute to better outcomes. \xwx{We further analyze whether the refinement trajectories in PMVisBench align with real user patterns. Participants required 3.32 rounds on average to converge to a final intent, close to the 3.7 rounds in PMVisBench. Across early/middle/late phases, visualization requirements are mainly introduced in the early or late stages (about 40\% each), data encoding refinements are more evenly distributed, and grouping/sorting refinements concentrate in the late stage (about 45\%). User refinements also show diverse intent-update combinations rather than a fixed order. These findings motivate our randomized clause-masking design, which avoids a single deterministic path and simulates uncertain progressive refinement, supporting the rationality of PMVisBench in capturing realistic refinement patterns and broad refinement types.}

\subsection{Case Study}

Tab.~\ref{tab:case_study} compares VQLs and visualizations produced by baselines and PMVisAgent for a representative NLQ in PMVisBench. Seq2Vis fails to generate an executable VQL due to join and column-reference errors. ncNet and MMCoVisNet also fail by referring to non-existent schema elements (e.g., an invalid table `type' or column `name'). Transformer generates an executable query (Fig.~\ref{fig:case_transformer}) but selects the wrong tables or columns, yielding incorrect results. RGVisNet retrieves relevant data but misses the required sorting (Fig.~\ref{fig:case_rgvisnet}). Prompt4Vis and nvAgent produce syntactically valid VQLs, yet introduce an extra aggregation on the y-axis, resulting in an incorrect bar chart (Fig.~\ref{fig:case_prompt4vis}). In contrast, PMVisAgent generates the correct VQL, matching the ground truth (Fig.~\ref{fig:case_gt}).

Figure~\ref{fig:case_study} demonstrates an example of the PMVisAgent workflow in round N. In this case, the Validation Agent identifies a schema-grounding error where the column \texttt{Address} does not exist and repairs it to \texttt{Street\_address}. It further detects intent issues, such as an unnecessary \texttt{COUNT(Floors)} aggregation and redundant ordering expressions, and refines the query. By leveraging the PMVis paradigm and validation, the PMVisAgent ensures VQL executability, ultimately converging to the accurate VQL.  This highlights the robustness of PMVisAgent in mitigating error accumulation and faithfully reproducing user-specified visualization requirements.

\section{Conclusion}

In this study, we introduced PMVis, a progressive multi-turn paradigm for text-to-visualization, which addresses the limitations of the traditional one-shot paradigm by enabling users to iteratively refine their analytical intents. To support this paradigm, we constructed PMVisBench, the first benchmark dataset specifically designed to capture the iterative nature of real-world user queries. Furthermore, building on PMVis, we developed PMVisAgent, an agent-based framework that integrates user–system dialogues with a clarification mechanism based on ReAct-style reasoning, effectively mitigating error accumulation across rounds.
Experiments on PMVisBench demonstrated that PMVisAgent significantly outperforms state-of-the-art baselines based on the one-shot paradigm. Ablation studies further validated the contributions of progressive interaction and clarification, while case studies illustrated how PMVisAgent generates accurate visualizations.

Limitations and Future Work.
While PMVisAgent demonstrates strong performance, its current design relies on multi-agent collaboration and iterative reasoning, which inevitably brings higher computational cost and lower efficiency compared to methods based on the one-shot paradigm. This may hinder real-time deployment in interactive analytics. In future work, we plan to explore lightweight agent coordination strategies and adaptive reasoning mechanisms to improve inference efficiency without sacrificing performance.

\bibliographystyle{IEEEtran}
\bibliography{ref}

\end{document}